%% file: main.tex
\documentclass[conference,a4paper]{IEEEtran}
\usepackage[utf8]{inputenc}
\usepackage{amsmath}
\usepackage{amsthm}
\usepackage{epsfig}          
\usepackage[left=0.8in,right=0.8in,top=1.1in,bottom=0.8in]{geometry}
\usepackage{textcomp}
\usepackage{times,latexsym,amsfonts,amsmath,amssymb,mathrsfs,verbatim,cite}
\newtheorem{theorem}{Theorem}
\newtheorem{definition}{Definition}

\newtheorem{lemma}{Lemma}

\usepackage{enumerate}
\usepackage{tikz}
\usepackage{epstopdf}
\usepackage{hyperref}
\usepackage{xspace}

\newcommand{\RT}{\ensuremath{\mathfrak{T}}\xspace}
\newcommand{\Rtens}[2]{\ensuremath{\RT({#1};{#2})}\xspace}
\newcommand{\ter}{\ensuremath{{\sf final}}}

\begin{document}

\sloppy

\title{A New Upperbound for the Oblivious Transfer Capacity of Discrete Memoryless Channels} 
  \author{
  \IEEEauthorblockN{K. Sankeerth Rao }
  \IEEEauthorblockA{Department of Electrical Engineering\\
                    Indian Institute of Technology, Bombay\\
                    Mumbai, India\\
                    Email: sankeerth1729@gmail.com} 
  \and
  \IEEEauthorblockN{Vinod M. Prabhakaran}
  \IEEEauthorblockA{School of Technology and Computer Science\\
                    Tata Institute of Fundamental Research\\
                    Mumbai, India\\
                    Email: vinodmp@tifr.res.in} 
  
   }

\maketitle

\begin{abstract}
We derive a new upper bound on the string oblivious transfer capacity of discrete memoryless channels (DMC). The main tool we use is the tension region of a pair of random variables introduced in Prabhakaran and Prabhakaran (2014) where it was used to derive upper bounds on rates of secure sampling in the source model. In this paper, we consider secure computation of string oblivious transfer in the channel model. Our bound is based on a monotonicity property of the tension region in the channel model. We show that our bound strictly improves upon the upper bound of Ahlswede and Csisz\'ar (2013).
\end{abstract}

\input{introduction.tex}
\input{problemstatement.tex}

\input{proof.tex}
\input{example.tex}
\input{discussion.tex}

\section*{Acknowledgment}

The research was funded in part by a grant from the Information Technology
Research Academy, Media Lab Asia, to IIT Bombay and TIFR, and by a
Ramanujan Fellowship from the Department of Science and Technology,
Government of India, to V. Prabhakaran. The first author would like to
thank Bharti Centre for Communications for supporting his research and
Manoj Mishra for all the help.

\input{bib.tex}
\appendix
\input{appendix.tex}
\end{document}

%% file: introduction.tex
\section{Introduction}

The goal of secure function computation is for users in a network to compute
functions of their collective data in such a way that users do not learn any
additional information about the data than the output of the functions they are
computing. This forms a central theme of modern cryptography under the rubric
of Secure Multiparty Computation.

In general, information theoretically secure function computation between
two users, who are equipped only with private/common randomness and
noiseless communication channels between them, is infeasible except for a
class of essentially trivial functions~\cite{Kushilevitz92}. However,
Cr\'epeau and Kilian showed that any function may be computed information
theoretically securely if a (non-trivial) noisy channel is available from
one of the users to the other~\cite{CrepeauKi88}. The approach was to show
that a certain primitive secure computation called {\em oblivious transfer}
(OT)~\cite{Wiesner88} is feasible given such a noisy channel resource, and
then rely on a reduction of of two-party computation to OT by
Kilian~\cite{Kilian88}.

OT (more specifically, 1-out-of-2 $m$-string OT) is the following secure
function computation between two users, say, Alice and Bob: Alice is given
2 strings $S_0, S_1$ picked independently and identically uniformly
distributed from $\{0,1\}^m$, Bob is given a uniform binary bit $K$,
independent of $S_0, S_1$. Alice is required to produce no output and Bob
should output $S_K$. Furthermore, Alice should not learn any information
about $K$ and Bob should not learn anything about the string $S_{\bar{K}}$,
where $\bar{K}=K + 1 \text{ mod } 2$. As must be clear from the discussion
above, OT cannot be securely computed when Alice and Bob only have access
to noise-free communication channels and private/common randomness.

Motivated by its role in secure computation, several works have addressed
the rate at which OT can be obtained from a discrete memoryless channel
(DMC). In~\cite{NascimentoW08}, OT capacity of a DMC was defined as the
largest rate of $m$-over-$n$, where $n$ is the number of channel uses,
achievable when Bob recovers $S_K$ with vanishing probability of error and
under vanishing information leakage measured via conditional mutual
informations. The paper also characterized noisy resources which provide a
strictly positive OT capacity. The OT capacity of erasure channels was
obtained in~\cite{ImaiMN06} for the honest-but-curious setting, where the
users follow the protocol faithfully, but attempt to derive information
they are not allowed to know from everything they have access to at the end
of the protocol. Ahlswede and Csisz\'ar~\cite{AhlswedeC13} characterized
the OT capacity for a more general class of channels called the {
generalized erasure channels}. In~\cite{PintoDMN11}, it was shown that the
OT capacity of generalized erasure channels remain the same even when the
users are allowed to be malicious.  The best known upper bounds on the OT
capacity of DMCs are due to Ahlswede and
Csisz\'ar~\cite{AhlswedeC13}\footnote{The same upper bounds can be inferred
from an earlier work by Wolf and Wullschleger~\cite{WolfW08} for the case
of zero-error and perfect privacy.}. These bounds, which apply for the case
of honest-but-curious users (and therefore, also for malicious users), were
obtained by weakening the problem of obtaining OT from a DMC to a secret
key agreement problem. In this paper we strictly improve upon these bounds. 

The main tool we use is the {\em tension region} $\Rtens{U}{V}$ of a pair
of random variables $U,V$ introduced in~\cite{PrabhakaranPr14}. Defined as
the increasing hull of the set of all $(I(V;Q|U),I(U;Q|V),I(U;V|Q)$, where
$Q$ is some random variables jointly distributed with $U,V$, it has the
interpretation as a rate-information tradeoff region for a distributed
common randomness generation problem which generalizes the setting of
G\'acs and K\"orner~\cite{GacsK73}. Specifically, consider a genie who has
access to $U^n,V^n\sim p(u,v)$ i.i.d., who needs to communicate to a user
with only $U^n$ and separately to a user with only $V^n$ such that two
users may agree (with vanishing probability of error) on a common random
variable $W$. The ``quality'' is measured by how small the average
``residual information'' $I(U^n;V^n|W)/n$ is. It was shown
in~\cite{PrabhakaranPr14} that the trade-off between the two rates of
communication from the genie to the users and the quality of the common
random variable agreed by the users is given by the tension region.

In~\cite{PrabhakaranPr14}, properties of tension region were used to derive
upper bounds on the rate of a form of secure computation with no inputs,
but randomized outputs, called {\em secure sampling} for the {\em source
model}, i.e., the ``noisy'' resource available to the two users are
observations from a distributed source (rather than a noisy channel as
here), and the goal of the secure computation is to produce samples of
another distributed source in such a way that neither user can infer any
more information about each other's output than can be inferred from their
own outputs\footnote{In fact, it is easy to show that secure computation of
OT is equivalent to secure sampling of the distribution: $A=(W_0,W_1)$ by
Alice and $B=(J,W_J)$ by Bob, where $W_0,W_1\in\{0,1\}^m$ and $J\in\{0,1\}$
indepedent and uniform. Hence, the results in~\cite{PrabhakaranPr14} can be
used to derive bounds on OT capacity of discrete memoryless sources. Using
a lemma of this paper, we give explicit bounds in
Section~\ref{sec:discussion}.}. The upper bound technique was a monotonicity
result for secure sampling protocols which implies that the tension region
of the outputs must contain the tension region of the distributed source
samples.

In contrast, this paper deals with the {\em channel model}. The main
technical contributions include a version of the monotonicity result for
the channel model. It turns out that, unlike in the source model, the whole
tension region does not satisfy a useful (i.e., single-letterizable)
monontonicty property, but its restriction to the $I(V;Q|U)=0$ plane does.
Specifically, we show that the restricted tension region of the
inputs-and-outputs of the function being securely computed must contain the
(Minkowski) sum of the restricted tension regions of the input and output
of the DMC at each channel use. We turn this into an upper bound on the OT
capacity by characterizing the restricted tension region of the
inputs-and-outputs of the OT function. In the interest of space, we only
present the argument required to obtain our upper bound on OT capacity in
this paper. The more general monotonicity result is deferred to a
full-length version.

%% file: problemstatement.tex
\section{Problem Statement}
\usetikzlibrary{decorations.markings}
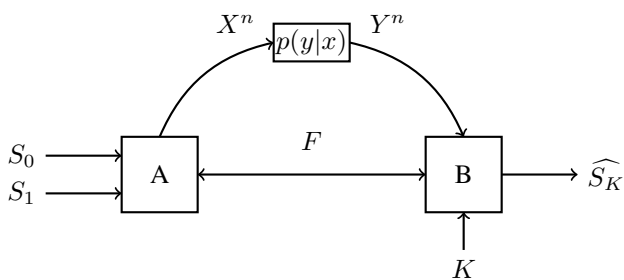
\begin{figure}[htb]
\setlength{\unitlength}{1cm}
\centering
\begin{tikzpicture}[scale=1, thick]
\draw (1,0) rectangle (2,1);
\draw (5,0) rectangle (6,1);
\draw (3,2) rectangle (4,2.5);

\draw [->] (0,0.25) -- (1,0.25);
\draw [->] (0,0.75) -- (1,0.75);
\draw [->] (5.5,-0.5) -- (5.5,0);
\draw [->] (6,0.5) -- (7,0.5);

\draw[->] (1.5,1) to [bend left=27] (3,2.25);
\draw[->] (4,2.25) to [bend left=27] (5.5,1);

\draw [<->] (2,0.5) -- (5,0.5);
\node at (3.5,2.25) {$p(y|x)$};
\node at (1.5,0.5) {A};
\node at (5.5,0.5) {B};

\node [below] at (3.5,1.2) {$F$};
\node [left] at (0,0.25) {$S_1$};
\node [left] at (0,0.75) {$S_0$};
\node [below] at (5.5,-0.5) {$K$};
\node[right] at (7,0.5) {$\widehat{S_K}$};
\node[above] at (2.5,2.25) {$X^n$};
\node[above] at (4.5,2.25) {$Y^n$};

\end{tikzpicture}
\caption{String Oblivious Transfer}
\label{otp}
\label{fig:fig1}
\end{figure}

Consider the setup in Fig.~\ref{fig:fig1}.  Alice's data are two strings
$S_0,S_1$ chosen independently and uniformly from $\{0,1\}^m$. Bob's data is a
uniform bit $K\in\{0,1\}$ independent of $S_0,S_1$. The goal is for Bob to
learn $S_K$. We require that neither user learn any (significant) amount of
additional information about the other's data apart from Bob learning $S_K$.
They have access to unlimited amounts of private randomness (i.e., randomness
independent of each other and of the data) and a noiseless discussion channel.
There is also a DMC from Alice to Bob: $p(y|x)$ where $ x \in\mathcal{X}$, the
input alphabet, and $y \in\mathcal{Y}$, the output alphabet. Before each
instance of using the DMC and after the last use of the DMC, Alice and Bob may
exchange messages with each other over the noiseless discussion channel,
potentially over multiple rounds. There are no constraints on the number of
rounds of message exchange they may engage in over the discussion channel
except that it be finite with probability 1.  We assume that the users are
honest-but-curious.

\begin{definition} Alice and Bob are said to have followed an {\em
$(n,m,\epsilon)$~secure protocol} if, the strings $S_0$ and $S_1$ input to
Alice have length $m$ each (as above), the protocol makes $n$ uses of the
DMC, and at the end of the protocol, Bob can output $\widehat{S_K}$ which
agrees with $S_K$ with probability at least $1-\epsilon$, and if the
transcript $F$ of the messages exchanged on the discussion channel and the
inputs $X^n$ and outputs $Y^n$ of the DMC satisfy the following privacy
constraints\footnote{Notice that we do not need to explicitly bring in the
private random variables in defining the privacy conditions since,
conditioned $F,X^n,S_0,S_1$, Alice's private randomness is independent of
$K$ and, similarly, conditioned on $F,Y^n,K$, Bob's private randomness is
independent of $S_0,S_1$.}: 
\begin{align}
I(F,X^n;K|S_0,S_1) &\leq \epsilon, \label{eq:security-against-A}\\
I(F,Y^n;S_{\bar{K}}|K) &\leq n\epsilon. \label{eq:security-against-B}
\end{align} 
\end{definition} 
Notice that \eqref{eq:security-against-A} guarantees Bob's privacy against
Alice, and \eqref{eq:security-against-B} guarantees privacy for Alice against
Bob.

\begin{definition}
A rate $R$ is said to be {\em achievable} if there is a sequence of
$(n,nR,\epsilon_n)$~secure protocols such that $\epsilon_n\to 0$ as $n \to
\infty$. The supremum of all achievable rates is the {OT~capacity},
$C$, of the DMC.
\end{definition} 

Our main result is the following upper bound on OT capacity. 
\begin{theorem}\label{bound}\label{thm:main}
\begin{align}
C \leq \max\limits_{p(x)} \min\limits_{Q-X-Y} I(X;Q|Y) + I(X;Y|Q),
\label{eq:main}
\end{align}
where the the minimization is over random variables $Q$ jointly distributed
with $X,Y$ satisfying the Markov chain constraint $Q-X-Y$ and the
cardinality bound $|\mathcal{Q}| \leq |\mathcal{X}||\mathcal{Y}| + 2$.
\end{theorem}

The currently best known upper bound is due to Ahlswede and
Csisz\'ar~\cite{AhlswedeC13}: 
\begin{align}
C \leq \max\limits_{p(x)} \; \min(I(X;Y), H(X|Y)).\label{eq:AC13}
\end{align}
It is easy to see that Theorem~\ref{bound} subsumes this. For a fixed
$p(x)$ in \eqref{eq:main}, notice that choosing $Q = \emptyset$ gives the
bound $I(X;Y)$, and choosing $Q=X$ gives the bound $H(X|Y)$. We shall show
in Section~\ref{ex} that our bound is a strict improvement on
\eqref{eq:AC13}.

%% file: proof.tex
\section{Proof of Theorem 1}{\label{prf}}

Conisder an $(n,nR,\epsilon)$ secure protocol, let all the random variables
that Alice has access to {\em after} the $i$-th usage of the DMC be called
the {\em view} of Alice at the $i$-th stage and be represented by $U_i$,
$i=0,1,2,\ldots,n$, where $U_{0}$ denotes Alice's view at the beginning of
the protocol, i.e., $U_0$ is made up of $S_0,S_1,$ and the private
randomness of Alice.  Similarly, we define the view of Bob till the $i$-th
stage and denote it by $V_i$. Let $U_\ter$ and $V_\ter$ be the views of
Alice and Bob at the termination of the protocol (after Bob outputs).
$U_\ter$ consists of $S_0,S_1$, the transcript $F$ of the discussion over
the noisefree channel, the inputs $X^n$ to the DMC and Alice's private
randomness. Similarly, $V_\ter$ comprises $K,F,Y^n,\widehat{S_K}$, and
Bob's private randomness.   

For a pair of jointly distributed random variables $U,V$, let us define the function $\alpha(U;V)$
\begin{align}
\alpha(U;V) := \min\limits_{Q-U-V} I(U;Q|V) + I(U;V|Q). \label{eq:alpha}
\end{align}
This is closely related to the tension region $\Rtens U V$ of a pair of
random variables $U,V$ of~\cite{PrabhakaranPr14}.  We recall from there the
definition of $\Rtens U V$:
\begin{align*}
&&\Rtens UV &= i\Big(\Big\{ \big(I(V;Q|U),I(U;Q|V),I(U;V|Q)\big):\\
&&&\qquad\qquad \quad  Q
\text{ jointly distributed with } U,V\Big\}\Big),
\end{align*}
where $i({\sf S})$ denotes the {\em increasing hull} of ${\sf
S}\subseteq {\mathbb R}_+^3$, defined as $i({\sf S}) = \{ s \in {\mathbb
R}_+^3: \exists s'\in {\sf S} \text{ s.t. } s\geq s' \}$. Thus, we have
\[ \alpha(U;V) = \min\{ s_2+s_3 : (0,s_2,s_3) \in \Rtens U V\}.\]
From~\cite[Theorems~2.3 and 2.4]{PrabhakaranPr14}, we know that \Rtens U V is a
closed, convex region and that, without loss of generality, we may assume the
cardinality bound $|\mathcal{Q}| \leq |\mathcal{X}||\mathcal{Y}| + 2$ on
the alphabet of $Q$ in the definition. This justifies the use of $\min$ and
the cardinality bound in \eqref{eq:main} as well as the use of $\min$ in
\eqref{eq:alpha}.

As we will prove later, $\alpha$ as a function of the two views satisfies the following properties:
\begin{enumerate}[(a)]
\item $\alpha(U_{i};V_{i}) \leq \alpha(U_{i-1};V_{i-1}) + \alpha(X_i;Y_i),\, i=1,\ldots,n.$\\
This means that $\alpha$ of the views can increase at most by
$\alpha(X_i;Y_i)$ between the $(i-1)$-th and the $i$-th uses of the DMC.
Specifically, we will see that no increase in $\alpha$ can come from the
discussion over the noiseless channel, and an increase of at most
$\alpha(X_i;Y_i)$ accrues from the $i$-th use of the DMC. This allows us to
upper bound the increase in $\alpha$ of the views as the protocol
progresses. 

\item $\alpha(U_0;V_0) = \alpha(S_0S_1;K) = 0,$\\
$\alpha(U_\ter;V_\ter) = \alpha(U_n;V_n).$\\
This means that $\alpha$ of the initial views is $0$, and the $\alpha$ of
the final views is the same as after the final use of the DMC. 

\item $\alpha(S_0S_1;KS_K) \leq \alpha(U_\ter;V_\ter) + n\delta(\epsilon),$\\
where $\delta(\epsilon)\to 0$ as $\epsilon\to 0$.
This means that $\alpha$ of the final views must be at least close to the
$\alpha$ of the inputs and (ideal) outputs of Alice and Bob for the OT
function being securely computed. 

\item $\alpha(S_0S_1;KS_K) = nR$\\
This means that $\alpha$ when applied to the inputs and (ideal) outputs of
Alice and Bob gives the length of the input strings to Alice.

\item $\alpha(X;Y)$ is a concave function of $p(x)$ for a fixed $p(y|x)$.
This justifies the use of $\max$ instead of $\sup$ in \eqref{eq:main}.
\end{enumerate}
Now applying $(a)$ recursively and using $(b)$, we  get
\begin{align*}
\alpha(U_\ter;V_\ter) \leq \sum\limits_{i=1}^n \alpha(X_i;Y_i).
\end{align*}
Using $(c)$ and $(d)$, we get
\begin{align*}
nR = \alpha(S_0S_1;KS_K) \leq \alpha(U_\ter;V_\ter) + n\delta(\epsilon).
\end{align*}
Thus, we have
\begin{align*}
R &\leq \frac{1}{n}\sum\limits_{i=1}^n \alpha(X_i;Y_i) + \delta(\epsilon)
  \leq \max_{p(x)} \alpha(X;Y) + \delta(\epsilon).
\end{align*}
Thus, we may conclude that $\max\limits_{p(x)}\, \alpha(X;Y)$ is an upper
bound on the OT-capacity for the DMC $p(y|x)$.

It only remains to prove (a)-(e). 
\vspace{.1in}

\noindent(a) Let $\widetilde{U}_i$ and $\widetilde{V}_i$ be the views of Alice and Bob right before the $i$-th use of the DMC. Then, $\widetilde{U}_i=(U_{i-1},\Delta F_{i-1},X_i)$ and $\widetilde{V}_i=(V_{i-1},\Delta F_{i-1})$, where $\Delta F_{i-1}$ is the transcript of the messages exchanged over the noiseless discussion channel after the $i-1$-th use of the DMC and before the $i$-th use. Note that $U_i=\widetilde{U}_i$ and $V_i=(\widetilde{V}_i,Y_i)$. The following can be inferred from~\cite[Theorem~5.4]{PrabhakaranPr14}:
\[\Rtens {\widetilde{U}_i} {\widetilde{V}_i} \supseteq \Rtens {U_{i-1}}{V_{i-1}},\]
i.e., the tension region of views cannot shrink during the discussion phase, or by Alice doing a private computation of $X_i$. Hence,
\[ \alpha({\widetilde{U}_i};{\widetilde{V}_i}) \leq \alpha({U_{i-1}},{V_{i-1}}).\]
In fact, the second line of property (b) also follows similarly, i.e., $\alpha(U_\ter;V_\ter) = \alpha(U_n;V_n)$.
Property (a) now follows from the following lemma which is proved in the appendix.
\begin{lemma} \label{lem:propa}
\[\alpha(U_i;V_i) \leq \alpha({\widetilde{U}_i};{\widetilde{V}_i}) + \alpha(X_i;Y_i).\]
\end{lemma}

\vspace{.1in}
\noindent(b) By choosing $Q$ to be a constant, $\alpha(U_0;V_0)=\alpha(S_0,S_1;K)=0$ follows. Proof of $\alpha(U_\ter;V_\ter) = \alpha(U_n;V_n)$ was already mentioned in (a).

\vspace{.1in}
\noindent(c)
For a pair of random variables $U,V$, and $0\leq \epsilon \leq H(V|U)$, we define
\[ \alpha_{\epsilon}(U;V) = \min_{I(Q;V|U)\leq \epsilon} I(U;Q|V) + I(U;V|Q).\]
Note that $\alpha(U;V)=\alpha_0(U;V)$. We will need the following property
(proved in the appendix using the fact that $\Rtens U V$ is
closed~\cite[Theorem~2.4]{PrabhakaranPr14}).
\begin{lemma} \label{lem:continuity}
For any pair of random variables $U,V$, the function $\alpha_\epsilon(U;V)$ is right continuous in $\epsilon$ at 0.
\end{lemma}
Property (c) now follows from the following lemma (also proved in the appendix):
\begin{lemma} \label{lem:propcadditional}
\begin{align*}
 \alpha_\epsilon(S_0S_1;KS_k) \leq \alpha(S_0S_1FX^n;KS_KFY^n) + n\delta_1(\epsilon),\\
 \alpha(S_0S_1FX^n; KS_KFY^n) \leq \alpha(U_\ter;V_\ter) + n\delta_2(\epsilon),
\end{align*}
where $\delta_1(\epsilon)\to 0$ and $\delta_2(\epsilon)\to 0$ as $\epsilon
\to 0$.
\end{lemma}
The proof of the first part relies on the privacy
conditions~\eqref{eq:security-against-A}-\eqref{eq:security-against-B}. The
second part uses $P(\widehat{S_K}\neq S_K) \leq \epsilon$.

\vspace{.1in}
\noindent(d) We prove the following lemma in the appendix.
\begin{lemma}{\label{lem:propd}}
$I(S_0S_1;KS_K|Q)+I(S_0S_1;Q|KS_K) \geq nR$ for all $Q-S_0S_1-KS_K$.
\end{lemma}
The property follows by noticing that equality is achieved by $Q=\emptyset$.

\vspace{.1in}
\noindent(e) For $Q-X-Y$,
\begin{align*}
&I(X;Q|Y) + I(X;Y|Q)\\ &= I(XY:Q) - I(Y;Q) + I(X;Y|Q)\\
 &= I(X;Q) - I(Y;Q) + I(X;Y|Q)\\
 &= H(Q|Y) - H(Q|X) + H(Y|Q) -H(Y|X).
\end{align*}
For fixed $p(q|x)$ and $p(y|x)$, the above expression is concave in $p(x)$
since $H(Q|X),H(Y|X)$ are linear in $p(x)$, and both $H(Q|Y),H(Y|Q)$ are
concave in $p(x)$; the latter can be shown, for instance, using the
convexity of relative entropy. i.e., for the DMC $p(y|x)$, if we define
\[ f_{p(q|x)}(p(x)) := I(X;Q|Y) + I(X;Y|Q),\]
where the mutual information terms are evaluated using
$p(x,y,q)=p(x)p(q|x)p(y|x)$, then, for $0\leq \lambda \leq 1$,
\begin{align*} \lambda f_{p(q|x)}(p_1(x)) &+ (1-\lambda)
f_{p(q|x)}(p_2(x))\\ &\quad\leq
f_{p(q|x)}(\lambda p_1(x) + (1-\lambda)p_2(x)).
\end{align*}
Property~(e) now follows from noticing that $\alpha(X;Y)=\min\limits_{p(q|x)}
f_{p(q|x)}(p(x))$.

%% file: example.tex
\section{An Example}{\label{ex}}
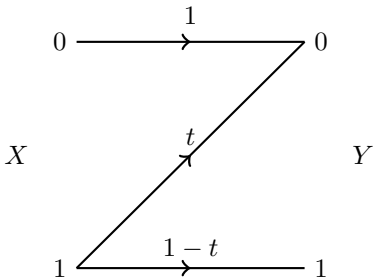
\begin{figure}[htb]
\setlength{\unitlength}{1cm}
\centering
\begin{tikzpicture}[scale=1, thick,decoration={ markings, mark=at position
0.5 with {\arrow[scale=1.3]{>}}}]

\draw[postaction={decorate}] (0,0) -- (3,0);
\draw[postaction={decorate}] (0,0) -- (3,3);
\draw[postaction={decorate}] (0,3) -- (3,3);

\node[above] at (1.5,0) {$1-t$};
\node[above] at (1.5,1.5) {$t$};
\node[above] at (1.5,3.1) {$1$};

\node[right] at (3,0) {$1$};
\node[right] at (3,3) {$0$};
\node[left] at (0,0) {$1$};
\node[left] at (0,3) {$0$};

\node[left] at (-.5,1.5) {$X$};
\node[right] at (3.5,1.5) {$Y$};

\end{tikzpicture}
\caption{The Z-channel (or binary asymmetric channel)}
\label{example}
\end{figure}

Consider the Z-channel $p(y|x)$ shown in Figure~\ref{example}.
$p(0|0)=1-p(1|0)=1$, and $p(0|1)=1-p(1|1)=t$, where $0\leq t\leq 1$.
Figure~\ref{fig:plot} compares the upper bound~\eqref{eq:main} on OT capacity
from Theorem~\ref{thm:main} with the upper bound~\eqref{eq:AC13} of Ahlswede
and Csisz\'ar~\cite{AhlswedeC13}. In fact, for ease of numerical calculation,
what is plotted is \eqref{eq:main} optimized over a smaller set of choices for
$Q$; specifically, we restrict to binary $Q$ and $p(q|x)$ of the form
$p(0|1)=0$. Even with this restriction, we observe that for a range of $t$'s
the upper bound of \eqref{eq:main} strictly improves upon \eqref{eq:AC13}.

For comparison, we also plot a simple lower bound to the OT capacity of this
channel. Let us consider two channel uses at a time. Now if we only use the
input letters from $\{01,10\}$, then this is a binary erasure channel (erasure
symbol 00) with erasure probability $t$ for which the OT capacity was shown
in~\cite{AhlswedeC13} to be $\min(1-t,t)$. So a lower bound for the OT capacity
of the Z-channel is $\frac{\min(1-t,t)}{2}$. We leave the problem of
characterizing the OT capacity of the Z-channel as an interesting open problem.
We conjecture that at least the lower bound, if not both the bounds, can be
improved.

\begin{figure}[tb]
\centering
\includegraphics[scale=0.45]{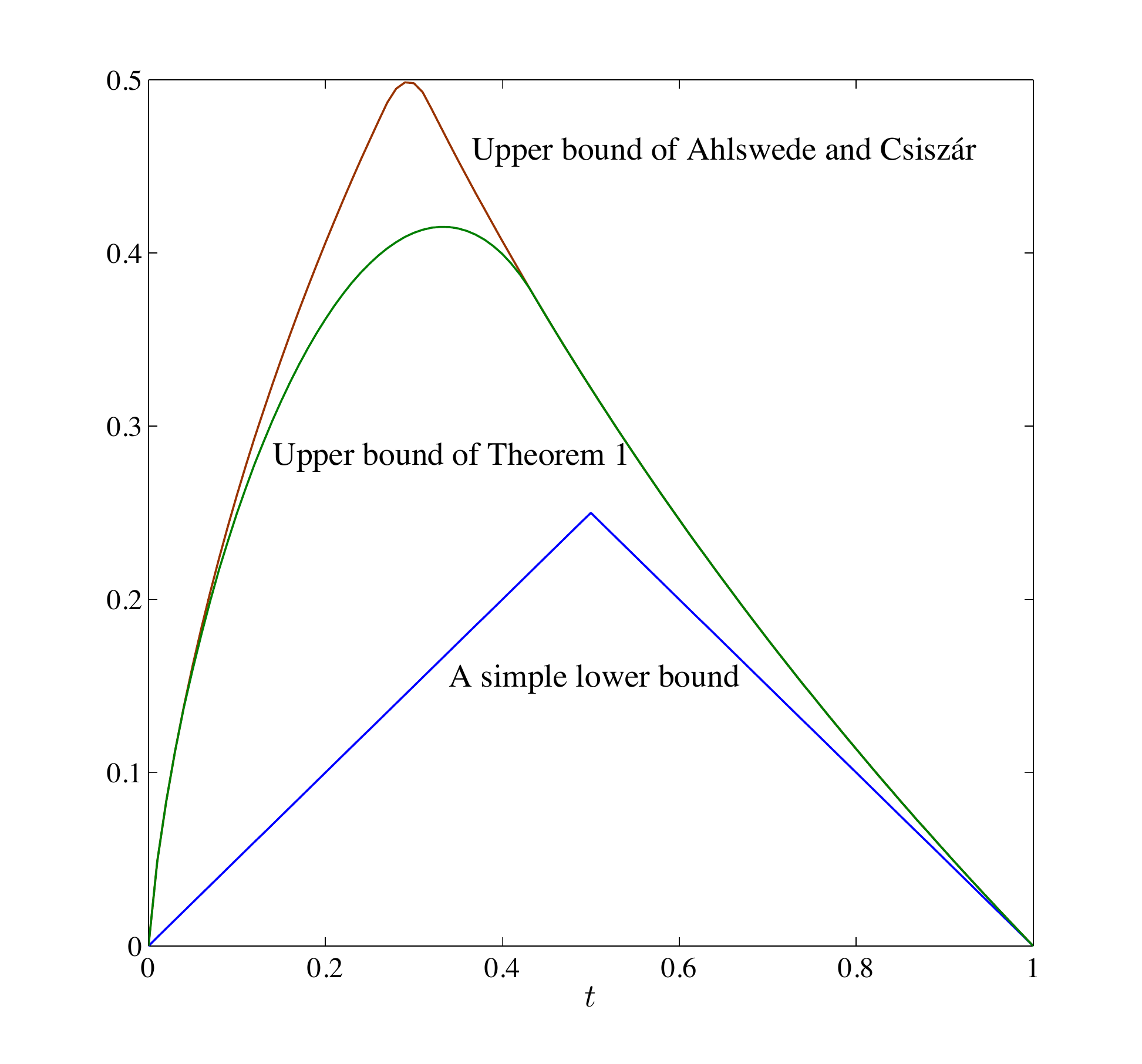}
\caption{Bounds on OT capacity of the Z-channel}
\label{fig:plot}
\end{figure}

%% file: discussion.tex
\section{Discussion} \label{sec:discussion}

An analogous upper bound on the OT capacity of the source model can be
derived using the results in~\cite{PrabhakaranPr14}. Applying
Lemma~\ref{lem:propd} of this paper to \cite[Corollary~5.8]{PrabhakaranPr14},
the OT capacity $C$ of the discrete memoryless source $p_{X,Y}$ can be
shown to satisfy 
\[ C\leq \min_{Q-X-Y} I(X;Q|Y) + I(X;Y|Q).\]
Details are deferred to a full-length version of this paper..

While this paper focused on deriving an upper bound on OT capacity of DMCs,
the technique is more general. In fact, we can derive a general upper bound
on the secure sampling capacity of DMCs analogous to the upper bound
in~\cite[Section~V]{PrabhakaranPr14} for the source model. The upper bound
on OT capacity presented here will follow as a corollary of such a general
upper bound. This is deferred to a full-length version.

The definition of OT capacity of DMCs
in~\cite{NascimentoW08},\cite{ImaiMN06},\cite{AhlswedeC13},\cite{PintoDMN11} is in
terms of the length of the string ($m$) per channel use. A different (not
equivalent) possibility is to fix $m$ (say $m=1$, for 1-bit OT) and consider
the number of independent $m$-string OTs obtained per channel use. This is of
interest since, in many secure computation protocols, several independent
instances of OT are called for (unlike the one instance of a long string-OT
considered in the original definition of OT capacity). We may also consider
varying the number of strings given to Alice and the number of strings picked
up Bob.  The general upper bound mentioned above provides means to derive upper
bounds on the rates in all these cases.

The achievability question of how to obtain ``high'' rates of secure
computation/sampling, in general, remains open. The capacity achieving
schemes for generalized erasure channels of~\cite{AhlswedeC13},\cite{PintoDMN11}
do not appear to extend to the general case. For the alternative
definitions of capacity mentioned above, the best achievability results
available for the general case only achieve very low (but non-zero)
rates~\cite{IshaiKORPSW11}. This is an important problem which requires
further research.

Unlike in the two-party setting, information theoretically secure computation,
in general, becomes feasible in the multiuser case even when only private
randomness at users and private noise-free channels between every pair of users
are available, provided the fraction of colluding adversarial users is
constrained (less than 1/2 for honest-but-curious and less than 1/3 for
malicious)~\cite{BGW88,CCD88}.  When such constraints are not satisfied,
availability of pairwise OTs, for instance, can enable secure computation in
general~\cite{HarnikIK07,PrabhakaranP12}.  Hence, OT capacity of multiuser
channels is also of interest~\cite{MishraDPD14}. Secure computation in
multiuser (noisy) networks is another question which merits further study.

%% file: appendix.tex
\begin{proof}[Proof of Lemma~\ref{lem:propa}]

Note that $U_i=\tilde{U}_i$ which contains $X_i$ as part of it, and $V_i=(\tilde{V}_i,Y_i)$.
Suppose we have $\tilde{Q}$ jointly distributed with $\tilde{U}_i,\tilde{V}_i$ such that $\tilde{Q} - \tilde{U}_i - \tilde{V}_i$ is a Markov chain, and $Q'$ is jointly distributed with $X_i,Y_i$ such that $Q' - X_i - Y_i$ a Markov chain. We define a random variable $Q$ with alphabet ${\mathcal Q} = \tilde{\mathcal Q} \times {\mathcal Q}'$, where $\tilde{\mathcal Q}$ and ${\mathcal Q}'$ are the alphabets of $\tilde{Q}$ and $Q'$ respectively, jointly distributed with $U_i,V_i$ as follows: 
\begin{align*}
p_{Q|U_i,V_i}((\tilde{q},q')|u_i,v_i) = p_{\tilde{Q}|\tilde{U}_i}(\tilde{q}|\tilde{u}_i)p_{Q'|X_i}(q'|x_i),
\end{align*}
where the $\tilde{u}_i$ on the right hand side is the same as $u_i$, and $x_i$ is the $x_i$ which is part of $u_i$. Notice that $Q - U_i - V_i$ is a Markov chain.

To prove the lemma, it is enough to show the following two inequalities
\begin{align}
I(U_i;Q|V_i) &\leq I(\tilde{U}_i;\tilde{Q}|\tilde{V}_i) + I(X_i;Q'|Y_i),\\
I(U_i;V_i|Q) &\leq I(\tilde{U}_i;\tilde{V}_i|\tilde{Q}) + I(X_i;Y_i|Q').
\end{align}
With come abuse of notation, if we write $Q=(\tilde{Q},Q')$, then 
\begin{align}
&p_{U_i,\tilde{V}_i,X_i,Y_i,Q}(u,\tilde{v},x,y,(\tilde{q},q'))\nonumber\\
&= p_{\tilde{U}_i,\tilde{V}_i}(u,\tilde{v})p_{\tilde{Q}|\tilde{U}_i}(\tilde{q}|u)p_{X_i|\tilde{U}_i}(x|u)p_{Y|X}(y|x)p_{Q'|X_i}(q'|x), \label{eq:propajoint}
\end{align}
where $p_{Y|X}$ is the DMC and $p_{X_i|\tilde{U}_i}$ is deterministic.
We have
\begin{align*}
I(U_i;Q|V_i) 
&= I(U_i;\tilde{Q}Q'|\tilde{V}_iY_i)\\
&= I(U_i;\tilde{Q}|\tilde{V}_iY_i) + I(U_i;Q'|\tilde{Q}\tilde{V}_iY_i)\\
&\leq I(U_iY_i;\tilde{Q}|\tilde{V}_i) + I(U_i\tilde{Q}\tilde{V}_i;Q'|Y_i)\\
&= [I(U_i;\tilde{Q}|\tilde{V}_i) + I(Y_i;\tilde{Q}|U_i,\tilde{V}_i)]\\ 
&\quad+
   [I(X_i;Q'|Y_i) + I(U_i\tilde{Q}\tilde{V}_i;Q'|X_iY_i)]\\
&= I(\tilde{U}_i;\tilde{Q}|\tilde{V}_i) + I(X_i;Q'|Y_i),
\end{align*}
where, in the penultimate step, we used the fact that $X_i$ is a part of
$U_i$, and in the last step, we used $U_i=\tilde{U}_i$ and the fact that, for
the joint distribution in \eqref{eq:propajoint}, $\tilde{Q} -
(U_i,\tilde{V}_i) - Y_i$ and $Q' - (X_i,Y_i) - (U_i,\tilde{Q},\tilde{V}_i)$
are Markov chains. 

Similarly,
\begin{align*}
I(U_i;V_i|Q)
&= I(U_i;\tilde{V}_iY_i|\tilde{Q}Q')\\
&= I(U_i;\tilde{V}_i|\tilde{Q}Q') + I(U_i;Y_i|\tilde{V}_i\tilde{Q}Q')\\
&\leq I(U_iQ';\tilde{V}_i|\tilde{Q}) + I(U_i\tilde{V}_i\tilde{Q};Y_i|Q')\\
&= [I(U_i;\tilde{V}_i|\tilde{Q}) + I(Q';\tilde{V}_i|U_i\tilde{Q})]\\
&\quad+
   [I(X_i;Y_i|Q') + I(U_i\tilde{V}_i\tilde{Q};Y_i|Q'X_i)]\\
&= I(\tilde{U}_i;\tilde{V}_i|\tilde{Q}) + I(X_i;Y_i|Q'),
\end{align*}
where the last step follows from the fact that, for the joint distribution
in \eqref{eq:propajoint}, $Q' - (U_i,\tilde{Q}) - \tilde{V}_i$ and
$(Q',Y_i) - X_i - (U_i,\tilde{V}_i,\tilde{Q})$ are Markov chains.
\end{proof}

\begin{proof}[Proof of Lemma~\ref{lem:continuity}]
We fix the joint distribution $U,V$. Below, we will write $\alpha_\epsilon$
to mean $\alpha_\epsilon(U;V)$.
Note that $\alpha_\epsilon$ is a non-increasing function of
$\epsilon$.
Suppose $\alpha_\epsilon$ is not (right) continuous at $\epsilon=0$, Then
there exists a sequence $\epsilon_n \to 0$
such that $\alpha_{\epsilon_n} \not\to \alpha_0$. So there exists a $\delta > 0$ and a
monotone subsequence $\epsilon'_n \downarrow 0$ such that
$\alpha_0-\alpha_{\epsilon'_n} \geq \delta, \forall
n$. Since $\alpha_{\epsilon'_n}$ is a monotone non-decreasing sequence bounded above it is
convergent. Let $l = \sup\limits_{n} \alpha_{\epsilon'_n}$. Then, $l = \lim\limits_{n \to \infty }
\alpha_{\epsilon'_n} \leq \alpha_0-\delta$. Since $\Rtens U V$ is a closed
region~\cite[Theorem~2.4]{PrabhakaranPr14}, so is
\[ \RT_{1,2+3}(U;V) := \{ (s_1,s_2+s_3): (s_1,s_2,s_3)\in\Rtens U V\}.\]
Hence, all the limit points of $\RT_{1,2+3}(U;V)$ lie in itself. So $(0,l) \in
\RT_{1,2+3}(U;V)$. This leads to a contradiction as $l \leq \alpha_0 -
\delta$ and, by definition, $\alpha_0$ is the minimum attainable value of
$s$ such $(0,s)\in\RT_{1,2+3}(U;V)$.
\end{proof}

\begin{proof}[Proof of Lemma~\ref{lem:propcadditional}]
The proof of the first part is along the lines of the proof of
property~3$'$
of~\cite[Theorem~5.7]{PrabhakaranPr14}. Consider any $Q$ jointly
distributed with $S_0,S_1,K,F,X^n,Y^n$. We have
\begin{align}
&I(KS_KY^nF;Q|S_0S_1FX^n)\nonumber\\
 &= I(KY^nF;Q|S_0S_1FX^n)\nonumber\\
 &\geq I(K;Q|S_0S_1FX^n)\nonumber\\
 &= I(K;QFX^n|S_0S_1) - I(K;FY^n|S_0S_1)\nonumber\\
 &\geq I(K;Q|S_0S_1) - I(K;FY^n|S_0S_1)\nonumber\\
 &\geq I(K;Q|S_0S_1) - \epsilon\quad\text{(by
\eqref{eq:security-against-A})}\nonumber\\
 &= I(KS_K;Q|S_0S_1) - \epsilon. \label{eq:pcadd1}\\
&I(S_0S_1X^nF;KS_KY^nF|Q) \geq I(S_0S_1;KS_K|Q).\label{eq:pcadd2}\\
&I(S_0S_1X^nF;KS_KY^nF;Q|KS_KY^nF)\nonumber\\ 
&= I(S_0S_1X^n;Q|KS_KY^nF)\nonumber\\
&\geq I(S_0S_1;Q|KS_KY^nF)\nonumber\\
&= I(S_0S_1;QY^nF|KS_K)-I(S_0S_1;Y^nF|KS_K)\nonumber\\
&\geq I(S_0S_1;Q|KS_K) - I(S_0S_1;Y^nF|KS_K). \label{eq:pcadd2.5}\\
\intertext{But,}
&I(S_0S_1;Y^nF|KS_K)\nonumber\\
&=I(S_{\bar{K}};Y^nF|KS_K)\nonumber\\
&=I(S_{\bar{K}};Y_nF|K) \quad\text{(by indep. of $S_0,S_1,K$)}\nonumber\\
&\leq n\epsilon. \quad\text{(by \eqref{eq:security-against-B})}
\label{eq:pcadd2.75}\\
\intertext{Substituting \eqref{eq:pcadd2.75} in \eqref{eq:pcadd2.5},}
&I(S_0S_1X^nF;KS_KY^nF;Q|KS_KY^nF)\nonumber\\
&\geq I(S_0S_1;Q|KS_K) - n\epsilon. \label{eq:pcadd3}
\end{align}
The first part of the lemma follows from
\eqref{eq:pcadd1},\eqref{eq:pcadd2}, and \eqref{eq:pcadd3}. Specifically,
\[ \alpha_\epsilon(S_0S_1;KS_K) \leq \alpha(S_0S_1X^nF;KS_KY^nF) +
n\epsilon.\]

To show the second part, let us first observe that $U_\ter$ contains $S_0S_1FX^n$ and $V_\ter$ contains $K\widehat{S_K}FY^n$. Further more,
\[ U_\ter - S_0S_1FX^n - K\widehat{S_K}FY^n - V_\ter\]
is a Markov chain, i.e., conditioned on $S_0,S_1,F,X^n$, Alice's final view (which only additionally contains her private randomness) is conditionally independent of Bob's final view, and similarly, Bob's view is conditionally independent of Alice's view conditioned on $K,\widehat{S_K},F,Y^n$. Hence, by property~3 of~\cite[Theorem~5.4]{PrabhakaranPr14}, we have
\[ \Rtens{S_0S_1FX^n}{K\widehat{S_K}FY^n} \supseteq \Rtens {U_\ter} {V_\ter}.\]
This implies that
\[ \alpha(S_0S_1FX^n;K\widehat{S_K}FY^n) \leq \alpha(U_\ter;V_\ter).\]
It remains to show that
\begin{align*}  \alpha(S_0S_1FX^n&;KS_KFY^n)\\ 
&\qquad\leq  \alpha(S_0S_1FX^n;K\widehat{S_K}FY^n) + n\delta_2(\epsilon).
\end{align*}

Let $U:=(S_0,S_1,F,X^n)$, $\hat{V}:=(K,\widehat{S_K},F,Y^n)$, and $V:=(K,S_K,F,Y^n)$. Since $F$ is part of both $U$ and $\hat{V}$, we have
\begin{align*}
 \alpha(U;\hat{V}) &= \min_{Q-U-\hat{V}} I(U;Q|\hat{V}) + I(U;\hat{V}|Q)\\
 &= \min_{Q'-U-\hat{V}} I(U;Q'F|\hat{V}) + I(U;\hat{V}|Q'F)
\end{align*}
Let $U':=(S_0,S_1,X^n)$, $\hat{V}':=(K,\widehat{S_K},Y^n)$, and $V':=(K,S_K,Y^n)$. Then,
\begin{align*}
\alpha(U;\hat{V}) = \min_{Q'-(U'F)-(\hat{V}'F)} I(U';Q'|\hat{V}'F) + I(U';\hat{V}'|Q'F).
\end{align*}
Similarly,
\begin{align*}
 &\alpha(U;V)\\ &\qquad = \min_{\tilde{Q}'-(U'F)-(V'F)} I(U';\tilde{Q}'|V'F) +
I(U';V'|\tilde{Q}'F).
\end{align*}
For $Q'$ jointly distributed with $(U',F,\hat{V}')$ such that
$Q'-(U',F)-(\hat{V}',F)$ is a Markov chain, we will define $\tilde{Q}'$
with the same alphabet as $Q'$ and jointly distributed with $(U',F,V')$
such that $Q'-(U',F)-(V',F)$ is a Markov chain by defining
\[ p_{\tilde{Q}'|U',F}(q'|u',f):=p_{Q'|U',F}(q'|u',f).\]
Then, since $P(\hat{V}\neq V)\leq \epsilon$, the total variation distance
between $(Q',U',F,\hat{V}')$ and $(\tilde{Q}',U',F,V')$ is at most $\epsilon$,
where total variation distance between two random variables $W$ and $W'$
defined over the same alphabet ${\mathcal W}$ is defined as
$\Delta(W,W')=\frac{1}{2}\sum_{w\in{\mathcal W}} |p_W(w)-p_{W'}(w)|$.

We will make use~\cite[Lemma~2.6]{PrabhakaranPr14} to obtain
\begin{align*}
I(U';\tilde{Q}'|V'F)&\leq I(U';Q'|\hat{V}'F)\\ &\qquad\qquad + 2H_2(\epsilon) +
\epsilon n(2R+\log|{\mathcal X}|)\\
I(U';V'|\tilde{Q}'F)&\leq I(U';\hat{V}'|Q'F)\\ &\qquad\qquad +
2H_2(\epsilon) + \epsilon n(2R+\log|{\mathcal X}|),
\end{align*}
where $H_2$ is the binary entropy function, and the term
$n(2R+\log|{\mathcal X}|)$ is, in fact, the cardinality of $U'$.
From this we may conclude that
\begin{align*}  \alpha(S_0S_1FX^n&;KS_KFY^n)\\ 
&\qquad\leq  \alpha(S_0S_1FX^n;K\widehat{S_K}FY^n) + n\delta_2(\epsilon),
\end{align*}
where $\delta_2(\epsilon)\to 0$ as $\epsilon \to 0$. This completes the
proof.
\end{proof}

\begin{proof}[Proof of Lemma~\ref{lem:propd}]
We have, 
\begin{align*}
&I(S_0S_1;Q|KS_K)\\
&= I(S_0S_1;QKS_K) - I(S_0S_1;KS_K)\\
&=I(S_0S_1;KS_K|Q)
+ I(S_0S_1;Q) - I(S_0S_1;KS_K).
\end{align*}
Using this, we can write
\begin{align}
&I(S_0S_1;KS_K|Q)+I(S_0S_1;Q|KS_K) \nonumber\\
&= 2I(S_0S_1;KS_K|Q) + [I(S_0S_1;Q) - I(S_0S_1;KS_K)]\nonumber\\
&= 2[H(KS_K|Q) - H(KS_K|S_0S_1)]\nonumber\\ &\hspace{2cm} + [H(S_0S_1|KS_K) - H(S_0S_1|Q)],
\label{eq:propodinter}
\end{align}
where in the last step we used the fact that $Q - S_0S_1 - KS_K$ is a Markov chain.
As we will argue below, under this Markov chain, 
\[ 2H(KS_K|Q) -  H(S_0S_1|Q) \geq 2.\]
Using this, along with $H(KS_K|S_0S_1)=1$ and $H(S_0S_1|KS_K)=nR$ in
\eqref{eq:propodinter}, we can conclude that
\[ I(S_0S_1;KS_K|Q)+I(S_0S_1;Q|KS_K) \geq nR.\]

It only remains to show that $2H(KS_K|Q) -  H(S_0S_1|Q) \geq 2$ if  $Q -
S_0S_1 - KS_K$ is a Markov chain. Since $K$ is independent of $(S_0,S_1)$, it is also independent of $(Q,S_0,S_1)$. Hence, using the fact that $K$ is a uniform bit,
\begin{align*}
2H(KS_K|Q)&=2H(K) + 2H(S_K|QK)\\
 &= 2 + H(S_0|Q,K=0) + H(S_1|Q,K=1)\\
 &= 2 + H(S_0|Q)+H(S_1|Q)\\
 &\geq 2 + H(S_0S_1|Q).
\end{align*}
This completes the proof.
\end{proof}